\documentclass{amsart}
\usepackage{color}
\usepackage{graphicx,epsf}
\usepackage{amsmath,amscd}
\usepackage{enumerate}
\newtheorem{theorem}{Theorem}[section]
\newtheorem{lemma}[theorem]{Lemma}

\newtheorem{proposition}[theorem]{Proposition}
\theoremstyle{definition}
\newtheorem{definition}[theorem]{Definition}

\theoremstyle{remark}
\newtheorem{remark}[theorem]{Remark}

\numberwithin{equation}{section}

\newcommand{\Real}{{\mathbb R}}

\newcommand{\eps}{\varepsilon}

\newcommand{\x}{\mathbf{x}}

\newcommand {\hide}[1]{}

\begin{document}
\title[Topological lower bounds for computation trees]
{On topological lower bounds for algebraic computation trees}
\thanks{Communicated by Felipe Cucker}
\thanks{2010 Mathematics Subject Classification 68Q17, 14P25}
\author{Andrei Gabrielov}
\address{Department of Mathematics,
Purdue University, West Lafayette, IN 47907, USA}
\email{agabriel@math.purdue.edu}
\author{Nicolai Vorobjov}
\address{
Department of Computer Science, University of Bath, Bath
BA2 7AY, England, UK}
\email{nnv@cs.bath.ac.uk}


\begin{abstract}
We prove that the height of any algebraic computation tree for deciding membership in a semialgebraic
set $\Sigma \subset \Real^n$ is bounded from below by
$$\frac{c_1\log ({\rm b}_m(\Sigma))}{m+1} -c_2n,$$
where ${\rm b}_m(\Sigma)$ is the $m$-th Betti number of $\Sigma$ with respect to ``ordinary'' (singular) homology,
and $c_1,\ c_2$ are some (absolute) positive constants.
This result complements the well known lower bound by Yao \cite{Yao} for {\em locally closed} semialgebraic sets
in terms of the total {\em Borel-Moore} Betti number.

We also prove that if $\rho:\> \Real^n \to \Real^{n-r}$ is the projection map, then the height of any tree
deciding membership in $\Sigma$ is bounded from below by
$$\frac{c_1\log ({\rm b}_m(\rho(\Sigma)))}{(m+1)^2} -\frac{c_2n}{m+1}$$
for some positive constants $c_1,\ c_2$.

We illustrate these general results by examples of lower complexity bounds for some specific
computational problems.
\end{abstract}
\maketitle

\section{Introduction}

The {\em algebraic computation tree} is a standard sequential model of computation for deciding membership problems
for semialgebraic sets.
Among various general methods for obtaining lower complexity bounds for this model, one of the most efficient
uses homotopy invariants, Euler characteristic and Betti numbers, as arguments for the bounding functions.
The history of this approach started probably in mid 70s with the work of Dobkin and Lipton \cite{DL}, and
features prominent results such as Ben-Or's \cite{Ben-Or} in 1983 and Yao's \cite{Yao} in 1997.
The present paper is inspired by the latter.
We will discuss the results of \cite{Yao} in some detail.

We assume that the reader is familiar with the concept of the algebraic computation tree, so we give here
just a brief formal description, closely following \cite{BC}.

\begin{definition}
An algebraic computation tree $\mathcal T$ with input variables $X_1, \ldots ,X_n$ taking real
values, is a tree having three types of vertices: computation (outdegree 1), branch (outdegree 3), and leaves
(outdegree 0).
To each vertex $v$ of $\mathcal T$ a variable $Y_v$ is assigned.

With each computation vertex $v$ an expression $Y_v = a \ast b$ is associated, where $\ast \in \{ +,-, \times, / \}$,
and $a, b$ are either real constants, or input variables, or
variables associated with predecessor vertices of $v$, or a combination of these.

At each branch vertex $v$, the variable $Y_v$ is assigned the value which is either a real constant, or an input variable,
or a variable associated with a predecessor vertex.
The three outgoing edges of $v$ correspond to signs $Y_v>0$, $Y_v=0$, $Y_v<0$.

With each leaf $w$ a basic semialgebraic set (called {\em leaf set}) is associated, defined by equations of the kind either
$Y_v >0$, or $Y_v=0$, or $Y_v <0$, for all variables $Y_v$ associated with predecessor branch vertices $v$ of $w$ along
the branch leading from the root to $w$.
The sign of each $Y_v$ is determined by the outgoing edge in the branch vertex.
In addition, each leaf carries a label ``Yes'' or ``No''.
The tree $\mathcal T$ {\em tests membership} in the union of all Yes leaf sets.
\end{definition}

The semantics of this model of computation is straightforward.
On an input $\x=(x_1, \ldots ,x_n) \in \Real^n$, the input variables $X_i$ get the corresponding values $x_i$,
the arithmetic operations are executed in computation vertices $v$ and the real values are obtained by variables
$Y_v$.
At branch vertices $v$ the sign of the value of $Y_v$ is determined, and the corresponding
outgoing edge is chosen.
As a result a certain branch ending up in a leaf $w$ is selected.
The input $\x$ belongs to the semialgebraic set assigned to $w$.
If $w$ is a Yes leaf, then $\x$ is said to be accepted by the tree $\mathcal T$.

It can be assumed without loss of generality \cite{BC}, that there are no divisions used in a tree.

We will be interested in lower bounds on the {\em heights} of algebraic computation trees testing
membership in a given semialgebraic set $\Sigma$.
A detailed outline of the development of lower bounds that depend on topological characteristics of a set
can be found in \cite{Yao} (see also \cite{BC}).
We mention here just two highlights.

The first most important achievement was the proof by Ben-Or \cite{Ben-Or} of the bound
$c_1\log ({\rm b}_0 (\Sigma))-c_2n$, where ${\rm b}_0(\Sigma)$ is the number of connected components of
$\Sigma$, and $c_1,\ c_2$ are some absolute positive constants.
This general bound implies non-trivial, and sometimes tight, lower bounds for specific computational
problems, such as {\em Distinctness} and {\em Knapsack}.

One of the most general results so far in this direction belongs to Yao \cite{Yao}.
Suppose a semialgebraic set $\Sigma$ is locally closed and bounded.
Let ${\rm b}^{BM}(\Sigma)$ be the total Betti number (the sum of all Betti numbers) of $\Sigma$ with respect
to the Borel-Moore homology $H^{BM}_\ast(\Sigma)$.
Yao proved the lower bound
\begin{equation}\label{eq:yao}
c_1\log ({\rm b}^{BM}(\Sigma))-c_2n,
\end{equation}
where $c_1,\ c_2$ are some absolute positive constants.
From this he deduced a tight lower bound for {\em $k$-Distinctness} problem, and other non-trivial lower
bounds for specific problems.

The Borel-Moore homology is a very strong condition, which implies subadditivity of the total Betti number.
Subadditivity is the property on which the whole of the Yao's argument depends.
It is natural to ask whether an analogous bound can be found for the usual, singular, homology theory, which is
applicable to arbitrary (not necessarily locally closed) semialgebraic sets.
Of course, in this case subadditivity is not necessarily true.
Observe that for compact sets Borel-Moore Betti numbers coincide with singular Betti numbers, while for
non-compact locally closed sets these two types of Betti numbers can be incomparable.

In this paper we prove two main theorems.
Firstly, we prove the lower bound
$$\frac{c_1\log ({\rm b}_m(\Sigma))}{m+1} -c_2n,$$
where ${\rm b}_m(\Sigma)$ is the $m$-th Betti number of an arbitrary semialgebraic set $\Sigma$ with respect
to singular homology, and $c_1,\ c_2$ are some absolute positive constants.
Note that this bound depends on an individual Betti number rather than on the sum of Betti numbers.
For Betti numbers of a small (fixed) index $m$ the bound turns into
$c_1 \log ({\rm b}_m(\Sigma)) -c_2n$ which is similar to Yao's bound.
The proof is based on a construction from \cite{GV09} which transforms $\Sigma$ into a compact semialgebraic
set $\Sigma'$ having the same Betti numbers as $\Sigma$ up to a given index $m$.
We then prove that for any algebraic computation tree $\mathcal T$ for $\Sigma$ there is an algebraic
computation tree ${\mathcal T}'$ for $\Sigma'$ having the height not exceeding, up to a multiplicative constant,
$m$ times the height of $\mathcal T$.
It remains to apply Yao's bound to $\Sigma'$.

Our second main result is a lower bound in terms of Betti numbers of the projection $\rho(\Sigma)$ of $\Sigma$
to a subspace, rather than Betti numbers of $\Sigma$ itself.
Note that the topology of the image under a projection may be much more complex than the topology of
the set being projected.
We are not aware of previous lower bounds of this sort.
The bound is
$$\frac{c_1\log ({\rm b}_m(\rho(\Sigma)))}{(m+1)^2} -\frac{c_2n}{m+1}$$
for some positive constants $c_1,\ c_2$, which again should be applied for small (fixed) values of $m$.
The proof uses (implicitly) a spectral sequence associated with the projection map, which allows to bound from above
Betti numbers of the projection of $\Sigma$ in terms of Betti numbers of fiber products by itself of
the compactification of $\Sigma$ \cite{GVZ}.

We illustrate these general results by examples of lower complexity bounds for some specific
computational problems.

\section{Topological tools}

In this section we formulate the results from \cite{GV05, GV09, GVZ} which are used further in this paper.

In what follows, for a topological space $X$, let ${\rm b}_m(X):={\rm rank}\ H_m(X)$ be its $m$-th Betti number
with respect to the singular homology group $H_m(X)$ with coefficients in some fixed Abelian group.
By ${\rm b}(X)$ we denote the {\em total} Betti number of $X$, i.e., the sum
$\sum_{i \ge 0} {\rm b}_i (X)$.

\subsection{Upper bounds on Betti numbers}
Consider a semialgebraic set $S=\{ \x \in \Real^n|\> {\mathcal F}(\x) \}$, where $\mathcal F$ is a Boolean
combination of polynomial equations and inequalities of the kind $h(\x)=0$ or $h(\x)>0$, and
$h \in \Real[x_1, \ldots ,x_n]$.
Suppose that the number of different polynomials $h$ is $s$ and that their degrees do not exceed $d$.

\begin{proposition}[\cite{GV09}, Theorem~6.3]\label{pr:bettibounds}
The $m$-th Betti number of $S$ satisfies
\begin{enumerate}
\item
${\rm b}_m(S)=O(s^2d)^n$;
\item
${\rm b}_m(S)=O((m+1)sd)^n$.
\end{enumerate}
\end{proposition}

\begin{remark}
Unlike classical Petrovski-Oleinik-Thom-Milnor bound for basic semialgebraic sets, used in \cite{Yao}, the bounds in
Proposition~\ref{pr:bettibounds} are applicable to arbitrary semialgebraic sets defined by a quantifier-free
formulae.
They are slightly weaker than the classical bound $O(sd)^n$ by a multiplicative factor at the base of
the exponent, namely $s$ in (1) and $m+1$ in (2).
\end{remark}

Further, in the proof of Theorem~\ref{th:generalbound}, we will need the bound (1) from this proposition.
We won't need bound (2) as such but we shall use in an essential way the constructions from \cite{GV09}
used for proving this bound.
We now proceed to describing this technique.

\subsection{Approximation by monotone families}\label{sub:approx}
\begin{definition}\label{def:S_delta}
Let $G$ be a compact semialgebraic set.
Consider a semialgebraic family $\{ S_\delta \}_{\delta >0}$ of
compact subsets of $G$, such that for all $\delta, \delta' \in (0,1)$,
if $\delta' > \delta$, then $S_{\delta'} \subset S_{\delta}$.
Denote $S := \bigcup_{\delta >0} S_{\delta}$.

For each $\delta >0$, let $\{ S_{\delta, \eps} \}$ be a semialgebraic family
of compact subsets of $G$ such that:
\begin{itemize}
\item[(i)]
for all $\eps, \eps' \in (0,1)$, if $\eps' > \eps$, then
$S_{ \delta, \eps} \subset S_{ \delta, \eps'}$;
\item[(ii)]
$S_{\delta}= \bigcap_{\eps >0} S_{\delta, \eps}$;
\item[(iii)]
for all sufficiently small $\delta' >0$ and for all $\eps' >0$, there exists an open in $G$ set $U \subset G$
such that $S_{\delta} \subset U \subset S_{\delta' , \eps'}$.
\end{itemize}
We say that $S$ is {\em represented} by the families $\{ S_\delta \}$ and
$\{ S_{\delta, \eps} \}$ in $G$.
\end{definition}

Consider the following two particular cases.
\medskip

\noindent {\bf Case 1.}\quad Let a semialgebraic set $S$ be given as a disjoint union of {\em basic} semialgebraic
sets (i.e., sets each defined by a conjunction of equations and strict inequalities).
(Note that an algebraic computation tree represents the corresponding set in exactly this way.)
Let $\delta$ and $\eps$ be some positive constants.

Suppose first that $S$ is bounded in $\Real^n$, and take as $G$ a closed ball of a sufficiently large
radius centered at 0.
The set $S_{\delta}$ is the result of the replacement, independently in each basic set in the union, of
all inequalities $h>0$ and $h<0$ by $h \ge \delta$ and $h \le -\delta$ respectively.
The set $S_{\delta, \eps}$ is obtained by replacing, independently in each basic set,
all expressions $h>0$, $h<0$ and $h=0$ by $h \ge \delta$, $h \le -\delta$ and
$h^2 - \eps \le 0$, respectively.
One can easily verify (see \cite{GV09}) that the set $S$, is represented by families $\{ S_{\delta} \}$ and
$\{ S_{\delta, \eps} \}$ in $G$.

Now suppose that $S$ is not necessarily bounded.
In this case one can take the semialgebraic one-point (Alexandrov) compactification of $\Real^n$ as $G$.
Define sets $S_\delta$ and $S_{\delta, \eps}$ as in the bounded case, replacing equations and
inequalities independently in each basic set,
and then taking the conjunction of the resulting formula with $|\x|^2 \le 1/\delta$.
Again, $S$ is represented by $\{ S_{\delta} \}$ and $\{ S_{\delta, \eps} \}$ in $G$.
\medskip

\noindent {\bf Case 2.}\quad Let $\rho:\> \Real^{n+r} \to \Real^n$ be the projection map, and
$S \subset \Real^{n+r}$ be a semialgebraic set, given as a disjoint union of {\em basic} semialgebraic sets.
The set $S$ is represented by families $\{ S_\delta \}$,
$\{ S_{\delta , \eps} \}$ in the compactification of $\Real^{n+r}$ as
described in {\bf Case~1}.
One can easily verify (see \cite{GV09}), that the projection $\rho (S)$ is represented by families
$\{ \rho(S_\delta) \}$, $\{ \rho(S_{\delta , \eps}) \}$ (in the Alexandrov compactification of $\Real^n$ if necessary).
\medskip

Returning to the general case, suppose that a semialgebraic set $S$ is {\em represented} by families
$\{ S_{\delta} \}$ and $\{ S_{\delta, \eps} \}$ in $G$.

For a sequence $\eps_0 , \delta_0 ,\eps_1 , \delta_1 , \ldots ,\eps_m , \delta_m$,
where $m \ge 0$, introduce the compact set
$$T_m(S):=S_{\delta_0,\eps_0}\cup S_{\delta_1,\eps_1} \cup \cdots \cup S_{\delta_m,\eps_m}.$$

Observe that in {\bf Case~2}, we have the equality
\begin{equation}\label{eq:rho}
T_m(\rho(S))=\rho(T_m(S)).
\end{equation}

In what follows, for two real numbers $a$ and $b$ we write $a \ll b$ to
mean ``$a$ is sufficiently smaller than $b$'' (see formal Definition~1.7 in \cite{GV09}).

\begin{proposition}[\cite{GV09}, Theorem~1.5]\label{pr:main}
For any $m \ge 0$, and
$$0<\eps_0 \ll\delta_0\ll\eps_1 \ll \delta_1 \ll \cdots \ll\eps_m \ll \delta_m \ll 1$$
we have
\begin{itemize}
\item[(i)]
for every $0 \le k \le m$, there is an epimorphism $\varphi_k:\> H_k(T_m(S)) \to H_k(S)$,
in particular, ${\rm b}_k(S) \le {\rm b}_k(T_m(S))$;
\item[(ii)]
in {\bf Case~1}, for every $k \le m-1$, the epimorphism $\varphi_k$ is an
isomorphism, in particular, ${\rm b}_k(S) = {\rm b}_k(T_m(S))$.
Moreover, if $m \ge \dim (S)$, then $T_m(S)$ is homotopy equivalent to $S$.
\end{itemize}
\end{proposition}

\subsection{Betti numbers of projections}

\begin{definition}
For two maps $f_1:\> X_1 \to Y$ and $f_2:\> X_2 \to Y$ , the {\em fibered product} of $X_1$ and $X_2$
is defined as
$$X_1 \times_Y X_2:=\{ (\x_1,\x_2)\in X_1 \times X_2|\> f_1(\x_1)=f_2(\x_2)\}.$$
\end{definition}

\begin{proposition}[\cite{GVZ}, Theorem~1]\label{pr:map}
Let $f:\> X \to Y$ be a closed surjective semialgebraic map (in particular, $f$ can be the projection
map to a subspace, with a compact $X$).
Then
$${\rm b}_m(Y) \le \sum_{p+q=m} {\rm b}_q(W_p),$$
where
$$W_p:= \underbrace{X \times_Y \cdots \times_Y X}_\text{(p+1) {\rm times}}.$$
\end{proposition}

\section{General lower bounds}

We start with a theorem which immediately follows from an upper bound on the total Betti number
of an arbitrary semialgebraic set in Proposition~\ref{pr:bettibounds}.

\begin{theorem}\label{th:generalbound}
Let $k$ be the height of an algebraic computation tree ${\mathcal T}$
testing membership in a semi-algebraic set $\Sigma \subset \Real^n$.
Then
$$
k = \Omega \left( \frac{\log ({\rm b}(\Sigma))}{n} \right),
$$
where $\rm b (\Sigma)$ is the total Betti number of $\Sigma$.
\end{theorem}

\begin{proof}
Since in each computation vertex at most one multiplication can be performed, every polynomial
occurring in the disjunctive normal form defining $\Sigma$ has a degree at most $2^k$.
The number of polynomials in the conjunction defining the set attached to a Yes leaf is at most $k$,
while the number of Yes leaves does not exceed $3^k$.
It follows that the total number of polynomials defining $\Sigma$ is at most $k3^{k}$.
Then, according to Proposition~\ref{pr:bettibounds}, (1), ${\rm b}(\Sigma) \le O(n((k3^{k})^2 2^k)^n)$.
Taking logarithms we get the result.
\end{proof}

\begin{remark}
The bound in the theorem is significantly weaker than Yao's bound (\ref{eq:yao}).
However, as explained in the introduction, it is applicable to any semialgebraic set, not necessarily
a locally closed one.
The upper bound on the total Betti number, used in the proof, is applicable to arbitrary
semialgebraic set, unlike classical Petrovski-Oleinik-Thom-Milnor bounds employed in \cite{Yao}.
\end{remark}

\begin{theorem}\label{th:main}
Let $k$ be the height of an algebraic computation tree ${\mathcal T}$
testing membership in a semi-algebraic set $\Sigma \subset \Real^n$.
Then
$$
k \ge \frac{c_1\log ({\rm b}_m(\Sigma))}{m+1} -c_2n,
$$
where ${\rm b}_m(\Sigma)$ is the $m$-th Betti number of $\Sigma$, and $c_1,\ c_2$ are some positive constants.
\end{theorem}

\begin{lemma}\label{le:cupcap}
Let ${\mathcal T}_1, \> {\mathcal T}_2$ be algebraic computation trees testing membership in semialgebraic
sets $\Sigma_1$ and $\Sigma_2$ respectively, and having heights $k_1$ and $k_2$ respectively.
Then there is a tree ${\mathcal T}_\cup$ testing membership in $\Sigma_1 \cup \Sigma_2$,
and a tree ${\mathcal T}_\cap$ testing membership in $\Sigma_1 \cap \Sigma_2$, both having heights
at most $O(k_1+k_2)$.
\end{lemma}

\begin{proof}
To construct ${\mathcal T}_\cup$, attach a copy of ${\mathcal T}_2$ to each No leaf of the tree ${\mathcal T}_1$.
For ${\mathcal T}_\cap$, attach a copy of ${\mathcal T}_2$ to each Yes leaf of the tree ${\mathcal T}_1$.
\end{proof}

\begin{lemma}\label{le:hight}
Let ${\mathcal T}$ be a tree for $\Sigma$, having height $k$.
Then for any $\ell \ge 0$ there exists a tree ${\mathcal T}_{\ell}$ for $T_{\ell}(\Sigma)$ whose height does not
exceed $k' \le c ((\ell+1) k +n)$ for some positive constant $c$.
\end{lemma}

\begin{proof}
The plan of the proof is as follows.
The construction of ${\mathcal T}_{\ell}$ consists of two stages.
On the first stage we perform the construction for $\ell=0$ and arbitrary $\eps, \delta$,
and get the tree ${\mathcal T}_{\eps, \delta}$.
The height of ${\mathcal T}_{\eps, \delta}$ is not larger than $c$ times the height of ${\mathcal T}$
for a constant $c >0$.
On the second stage we construct ${\mathcal T}_{\ell}$ for an arbitrary $\ell$ by induction.
On the base step, construct the tree ${\mathcal T}_0={\mathcal T}_{\eps_0, \delta_0}$.
Suppose we constructed the tree ${\mathcal T}_{\ell -1}$.
The tree ${\mathcal T}_{\ell}$ is obtained from ${\mathcal T}_{\ell -1}$ by attaching
to {\em each} No leaf of the latter, a copy of the tree ${\mathcal T}_{\eps_\ell, \delta_\ell}$, considering
the leaf as the root of ${\mathcal T}_{\eps_\ell, \delta_\ell}$.

\begin{figure}[hbt]
\vspace*{0.05in}
       \centerline{
          \scalebox{0.35}{
             \includegraphics{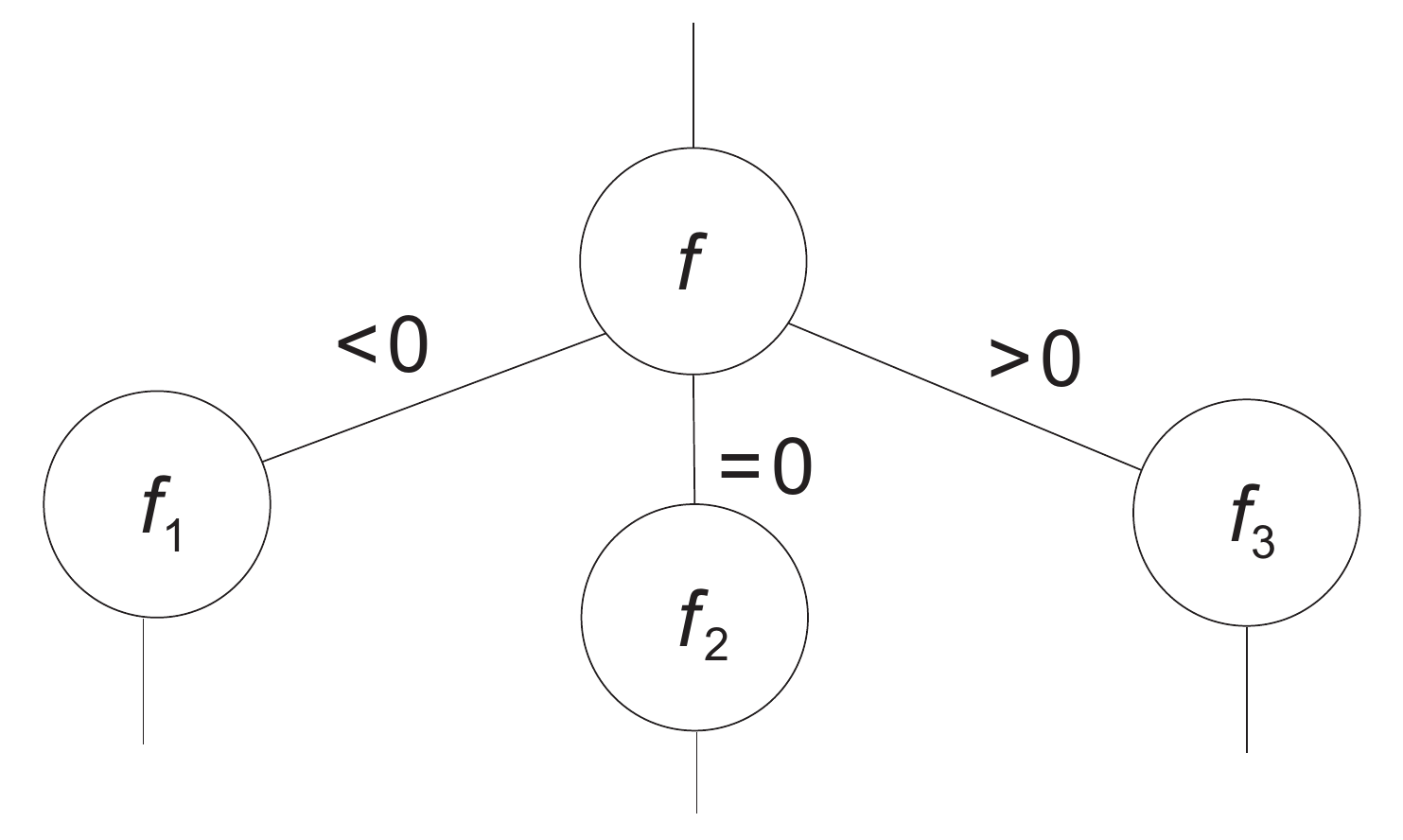}
             }
           }
\vspace*{-0.1in}
\caption{Tree $\mathcal S$.}
\label{fig:1}
\end{figure}

Now we proceed to a more detailed proof.

Let $h:=X_1^2+ \cdots +X_n^2$.
The root $r$ of the tree ${\mathcal T}_{\eps, \delta}$ is a branch vertex with the polynomial
$h- 1/\delta$ attached.
The child of $r$, corresponding to $>0$, is a No leaf.
Take the other two children as roots of two copies of the tree $\mathcal T$.
The construction of ${\mathcal T}_{\eps, \delta}$ now continues identically for both copies, by induction,
as follows.
In ${\mathcal T}$, let $v$ be the closest branch vertex to its root, and let $f$ be the polynomial
attached to $v$ (if such branch vertex does not exist, then  the construction of ${\mathcal T}_{\eps, \delta}$ is completed).
Then the neighbourhood of $v$ in ${\mathcal T}$ looks like the tree $\mathcal S$ on Figure~\ref{fig:1}.
Here $f_1, f_2$ and $f_3$ are polynomials attached to children $v_1,v_2,v_3$ of $v$.
Replace this neighbourhood by the tree ${\mathcal S}'$ on Figure~\ref{fig:2}.
Notice that the leaves of ${\mathcal S}'$, are labelled again by $f_1, f_2, f_3$ while one of the
leaves is a No leaf.
Attach to each leaf of ${\mathcal S}'$, labelled by $f_i$, the subtree of $\mathcal T$ rooted at $v_i$ (unless $v_i$ is a leaf of
$\mathcal T$).
Denote the resulting tree by ${\mathcal T}'$.
This completes the base of the induction.

\begin{figure}[hbt]
\vspace*{0.05in}
       \centerline{
          \scalebox{0.35}{
             \includegraphics{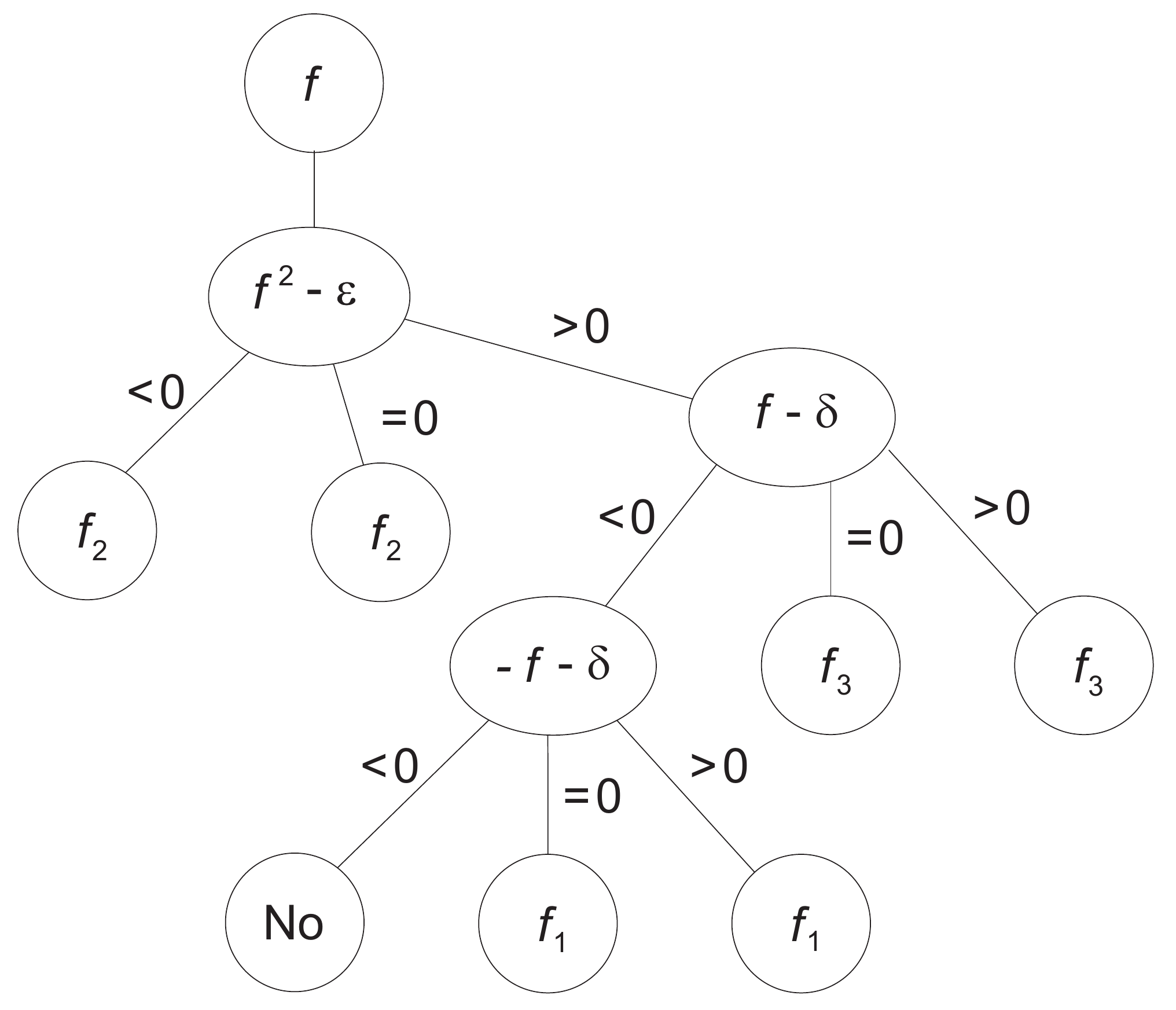}
             }
           }
\vspace*{-0.1in}
\caption{Tree ${\mathcal S}'$.}
\label{fig:2}
\end{figure}

On the next induction step perform the same replacement operation, as on the base step, for each subtree of
${\mathcal T}'$ rooted at a leaf of ${\mathcal S}'$ which is not a No leaf.
If for a leaf of ${\mathcal S}'$ no such subtree exists, i.e., vertex $v_i$ is a leaf of $\mathcal T$,
then this vertex is taken as a leaf of ${\mathcal T}_{\eps, \delta}$,
it is a Yes leaf if and only if $v_i$ is a Yes leaf in $\mathcal T$.
Denote the results of replacements again by ${\mathcal S}'$, and the resulting tree again by ${\mathcal T}'$.

Further induction steps are performed in the same fashion, by applying the replacement operation, described at the base step,
to subtrees of ${\mathcal T}'$ rooted at leaves of the trees ${\mathcal S}'$ obtained on the previous induction step.
The construction of ${\mathcal T}_{\eps, \delta}$ is completed when all leaves of trees ${\mathcal S}'$
become leaves of ${\mathcal T}_{\eps, \delta}$.

Note that the height of ${\mathcal T}_{\eps, \delta}$ is not larger than $c$ times the height of
${\mathcal T}$ for a constant $c>0$.

Now we prove by induction on the construction that ${\mathcal T}_{\eps, \delta}$ is a tree testing
membership in $\Sigma_{\eps, \delta}$ (recall the notation from Case~1, Section~\ref{sub:approx}).
Observe that either $h - 1/\delta=0$ or $h- 1/\delta<0$ is present in the definition of any Yes leaf set.
Assume, as before, that in ${\mathcal T}$ the vertex $v$ is the closest branch vertex to the root,
and $f$ is the polynomial attached to $v$.
Observe that each leaf set of ${\mathcal T}$, in particular each Yes leaf set, is of the kind either
$\{ f=0, \ldots \}$, or $\{ f>0, \ldots \}$, or $\{ f<0, \ldots \}$.
In the tree ${\mathcal T}'$, on the base step of the construction of ${\mathcal T}_{\eps, \delta}$,
the leaf $\{ f=0, \ldots \}$ will be replaced by two leaves, $\{ f^2- \eps=0, \ldots \}$ and
$\{ f^2- \eps<0, \ldots \}$, the leaf $\{ f>0, \ldots \}$ -- by two leaves, $\{ f- \delta=0, \ldots \}$ and
$\{ f- \delta >0, \ldots \}$, while the leaf $\{ f<0, \ldots \}$ -- by two leaves, $\{ -f- \delta=0, \ldots \}$
and $\{ -f- \delta >0, \ldots \}$.
It follows that $\{ f=0, \ldots \} \subset \Sigma$ if and only if $\{ f^2- \eps \le 0, \ldots \}$ is a subset of
the set tested by ${\mathcal T}'$, and similar for sets $\{ f>0, \ldots \}$ and $\{ f<0, \ldots \}$.
Proceeding by induction, we conclude that $\Sigma_{\eps, \delta}$ is the set tested by
${\mathcal T}_{\eps, \delta}$.

Now construct ${\mathcal T}_{\ell}$ for arbitrary $\ell$ by induction.
On the base step, start with the path of $O(n)$ computation vertices at the end
of which the polynomial $h=X^2_1+ \cdots +X^2_n$ is computed.
Continue with the tree ${\mathcal T}_{\eps_0, \delta_0}$.
The result of these two steps is the tree ${\mathcal T}_0$.
Suppose we constructed the tree ${\mathcal T}_{\ell -1}$ for $\ell \ge 1$.
The tree ${\mathcal T}_{\ell}$ is obtained from ${\mathcal T}_{\ell -1}$ by attaching
to {\em each} No leaf of the latter, the tree ${\mathcal T}_{\eps_\ell, \delta_\ell}$, considering the leaf
as the root of ${\mathcal T}_{\eps_\ell, \delta_\ell}$.
By Lemma~\ref{le:cupcap}, the result is indeed ${\mathcal T}_{\ell}$.

Obviously the height of ${\mathcal T}_{\ell}$ does not exceed $k' \le c ((\ell +1)k +n)$ for a constant $c>0$.
\end{proof}

\begin{proof}[Proof of Theorem~\ref{th:main}]
By Lemma~\ref{le:hight}, $k' \le c((m+1)k+n)$ for a constant $c>0$.
By (\ref{eq:yao}), since $T_{m}(\Sigma)$ is compact,
$$c((m+1)k+n) \ge c_1\log({\rm b}((T_{m}(\Sigma))))-c_2n$$
for some positive constants $c_1,\ c_2$.
Hence, for the $m$-th Betti number,
$$c((m+1)k+n) \ge c_1\log({\rm b}_m((T_{m}(\Sigma))))-c_2n.$$
It follows, by Proposition~\ref{pr:main}, that
$$c((m+1)k+n) \ge c_1\log({\rm b}_m(\Sigma))-c_2n.$$
Hence, the theorem.
\end{proof}

\section{Projections}

\begin{theorem}\label{th:proj}
Let $k$ be the height of an algebraic computation tree ${\mathcal T}$
testing membership in a semi-algebraic set $\Sigma \subset \Real^n$.
Let $\rho:\> \Real^n \to  \Real^{n-r}$ be the projection map.
Then
\begin{equation}
k \ge  \frac{c_1\log ({\rm b}_m(\rho(\Sigma)))}{(m+1)^2} -\frac{c_2n}{m+1}
\end{equation}
for some positive constants $c_1,\ c_2$.
\end{theorem}

Let
$$
W_p:=\underbrace{T_m(\Sigma) \times_{\rho(T_m(\Sigma))} \cdots \times_{\rho(T_m(\Sigma))}  T_m(\Sigma)
}_\text{(p+1) {\rm times}}.
$$

\begin{lemma}\label{le:proj}
Let ${\mathcal T}$ be a tree for $\Sigma$, having height $k$.
Then there exists a tree ${\mathcal T}^W_{m}$ for $W_p$ whose height does not
exceed $ c (p+1)((m+1)k +n)$ for some positive constant $c$.
\end{lemma}

\begin{proof}
Lemma~\ref{le:hight} implies that there is a tree ${\mathcal T}_m$ for $T_m(\Sigma)$ having the height not
exceeding $c((m+1) k+n)$.
The problem of membership in $W_p$ has input variables
$$X_1, \ldots ,X_{n-r}, Y_{1, n-r+1}, \ldots , Y_{1, n}, \ldots , Y_{p, n-r+1}, \ldots , Y_{p, n}.$$

Construct the tree ${\mathcal T}^W_{m}$ inductively, starting with a copy of ${\mathcal T}_m$ with
input variables $X_1, \ldots ,X_{n-r}, Y_{1, n-r+1}, \ldots , Y_{1, n}$.
Then, using Lemma~\ref{le:cupcap}, attach to each Yes leaf another copy of ${\mathcal T}_m$ with input variables
$X_1, \ldots ,X_{n-r}, Y_{2, n-r+1}, \ldots , Y_{2, n}$, and so on.
The height of the resulting tree ${\mathcal T}^W_{m}$ is at most $p+1$ times the height of the tree
${\mathcal T}_m$, as required.
\end{proof}

\begin{proof}[Proof of Theorem~\ref{th:proj}]
According to Proposition~\ref{pr:map},
\begin{equation}\label{eq:proj}
{\rm b}_m(\rho(T_m(\Sigma))) \le \sum_{p+q=m}{\rm b}_q(W_p).
\end{equation}
Let ${\rm b}(W_\nu):= \max_{0 \le p \le m}{\rm b}(W_p)$, and $k'$ be the height of a tree for $W_\nu$.

Since $W_\nu$ is compact, by (\ref{eq:yao}), we have
$$k' \ge c'_1 \log ({\rm b}(W_\nu)) - c'_2(n+\nu r)$$
for some positive constants $c'_1,\ c'_2$, thus, replacing $k'$ by a larger number according to Lemma~\ref{le:proj},
and using $m \ge \nu$, we get
$$(m+1)^2k+(m+1)n \ge c''_1 \log({\rm b}(W_\nu))- c''_2(n+(m+1)r)$$
for some positive constants $c''_1,\ c''_2$.
But
$$
\sum_{p+q=m}{\rm b}_q(W_p) \le m\ {\rm b}(W_\nu),
$$
so using (\ref{eq:proj}) we have
$$(m+1)^2k+(m+1)n \ge c''_1 (\log({\rm b}_m(\rho(T_m(\Sigma)))- \log m)- c''_2(n+(m+1)r).$$
Hence,
$$k \ge \frac{c_1 \log({\rm b}_m(\rho(T_m(\Sigma)))}{(m+1)^2}-\frac{c_2n}{m+1}$$
for some positive constants $c_1,\ c_2$.

According to (\ref{eq:rho}), $\rho(T_m(\Sigma))=T_m(\rho(\Sigma))$, while, by Proposition~\ref{pr:main},
$${\rm b}_m(T_m(\rho(\Sigma))) \ge {\rm b}_m(\rho(\Sigma)).$$
It follows that
$$k \ge \frac{c_1 \log({\rm b}_m(\rho(\Sigma)))}{(m+1)^2}-\frac{c_2n}{m+1}.$$
\end{proof}

\section{Applications}

In this section we apply the general bounds from Theorems~\ref{th:main} and \ref{th:proj} to examples of
specific computational problems.
These problems admit obvious variations.

\subsection*{``Parity of integers''}

This is the following computational problem.
\medskip

{\em Let $m$ be a positive integer.
For  given $n$ real numbers $x_1, \ldots ,x_n$ such that $1 \le x_i \le m$ for all $i$,
decide whether the following property is true: either all $x_i$ are integer or exactly two
of them are not integer.}
\medskip

Observe that complexity of this problem has an {\em upper bound} $O(n \log m)$:
the computation tree for each $x_i$ checks whether it coincides with one of the numbers $1, \ldots ,m$ using
binary search.

To obtain a lower bound, consider the integer lattice $\{1, \ldots, m\}^n$ in $\Real^n$ and let
$\Sigma$ be the union of all open 2-dimensional squares and all vertices.
Then the problem is equivalent to deciding membership in $\Sigma$.
Observe that $\Sigma$ is not locally closed.
It is homotopy equivalent to a 2-plane with $\Omega ( m^n)$ punctured points,
so $b_1(\Sigma) = \Omega ( m^n)$.
By Theorem~\ref{th:main}, the height of any algebraic computation tree testing membership in $\Sigma$
is $\Omega (n \log m)$.

\subsection*{``Crossing number''}
Let $\Sigma$ be a smooth connected bounded semialgebraic curve in $\Real^n$.
Then $\Sigma$ is a (smooth) embedding of either the circle $S^1$ or the interval $(0,1)$ into $\Real^n$
(for $n=3$ and a circle this is a knot).
The total Betti number of $\Sigma$ is at most 2.

Observe that the image under the projection of $\Sigma$ onto a generic 2-dimensional linear subspace has
only double points as singular points.

The {\em crossing number} $C(\Sigma)$ of $\Sigma$ is the maximal number of singular points of the image of the
projection over all generic 2-dimensional linear subspaces.

\begin{theorem}
The complexity of membership in $\Sigma$ is at least $c_1 \log C(\Sigma)-c_2n$ for some positive constants
$c_1,\ c_2$.

\end{theorem}

\begin{proof}
Let $\rho (\Sigma)$ be the image of $\Sigma$ under the projection to the plane on which the crossing number
is realized.
Then $C (\Sigma)$ is less by 2 (if $\Sigma$ is an embedding of $S^1$), or otherwise by 1, than the
number of connected components of the complement to $\rho (\Sigma)$ in the plane.
By Alexander duality, the number of connected components is the same as ${\rm b}_1(\rho( \Sigma))$,
hence the lower bound follows from Theorem~\ref{th:proj}.
\end{proof}

\subsection*{Acknowledgements}
Andrei Gabrielov was partially supported by NSF grant DMS-1161629.

\end{document}